\documentclass[final,3p,times]{elsarticle}

\usepackage{amsmath,amssymb,amsfonts,amsthm}
\usepackage{color}
\usepackage{enumerate}
\usepackage{algorithmic,algorithm}
\usepackage{graphics, graphicx}
\usepackage{float}
\usepackage{caption}
\usepackage{subfigure}

\newtheorem{thm}{Theorem}[section]
\newtheorem{prop}[thm]{Proposition}
\newtheorem{lem}[thm]{Lemma}
\newtheorem{defi}[thm]{Definition}
\newtheorem{rem}{Remark}

\newcommand{\skel}{\mathbb S}

\renewcommand\simeq{\cong}

\begin{document}

\sloppy

\begin{frontmatter}


\title{On the Cartesian Skeleton and the Factorization of the Strong Product of Digraphs}

\author[SB]{Marc Hellmuth\corref{cor1}}
				\ead{marc.hellmuth@bioinf.uni-sb.de}
				\cortext[cor1]{corresponding author}
\author[SB,LJ]{Tilen Marc}
				\ead{marct15@gmail.com}

\address[SB]{Center for Bioinformatics, Saarland University, Building E
             2.1, Room 413, P.O. Box 15 11 50, D-66041 Saarbr\"{u}cken,
             Germany
   					}
\address[LJ]{Faculty of Mathematics and Physics, University of Ljubljana
						 Jadranska 19, 1000 Ljubljana, Slovenia}

\begin{abstract}
The three standard products (the Cartesian, the direct and the strong
product) of undirected graphs have been well-investigated, unique
prime factor decomposition (PFD) are known and polynomial time
algorithms have been established for determining the prime factors. 

For directed graphs, unique PFD results with respect to the standard products
are known. However, there is still a lack of algorithms, that computes
the PFD of directed graphs with respect to the direct and the strong product in
general. In this contribution, we focus on the algorithmic aspects for
determining the PFD of directed graphs with respect to the strong product.
Essential for computing the prime factors is the construction of a so-called
Cartesian skeleton. This article
introduces the notion of the Cartesian skeleton of directed graphs as
a generalization of the Cartesian skeleton of undirected graphs. We
provide new, fast and transparent algorithms for its construction.
Moreover, we present a first polynomial time algorithm for determining
the PFD with respect to the strong product of arbitrary connected digraphs. 
\end{abstract}

\end{frontmatter}

\section{Introduction}

Graphs and in particular graph products arise in a variety of different contexts, 
from computer science \cite{AMA-07, JHH+10} to theoretical biology 
\cite{SS04,Wagner:03a}, 
computational engineering  \cite{ka13,KK-08, KR-04} 
or just as natural structures in discrete mathematics 
\cite{Hammack:2011a,IMKL-00}.

For undirected simple graphs, it is well-known that each of the three
standard graph products, the Cartesian product
\cite{FHS-85,impe-2007,SA-60,VI-66}, the direct product
\cite{IM-98,MCKE-71} and the strong product
\cite{DOIM-70,FESC-92,MCKE-71}, satisfies the unique prime factor
decomposition property under certain conditions, and there are
polynomial-time algorithms to determine the prime factors. Several
monographs cover the topic in substantial detail and serve as standard
references \cite{Hammack:2011a,IMKL-00}.

For directed graphs, or digraphs for short, only partial results are
known. Feigenbaum showed that the Cartesian product of digraphs 
satisfies the unique prime factorization property
and provided a polynomial-time algorithm for its
computation \cite{Fei86}. McKenzie proved that digraphs have a unique
prime factor decomposition w.r.t.\  direct product requiring strong conditions on connectedness
\cite{MCKE-71}. This result was extended by Imrich and Kl{\"o}ckl in
\cite{Kloeckl:07, Kloeckl:10}. The authors provided unique prime factorization
theorems and a polynomial-time algorithm for the direct product of
digraphs under relaxed connectivity, but additional so-called thinness
conditions. The results of McKenzie also  imply that the strong
product of digraphs can be uniquely decomposed into prime factors \cite{MCKE-71}. 
Surprisingly, so far no general algorithm for determining the prime factors
of the strong product of digraphs has been established.

In this contribution, we are concerned with the algorithmic aspect of
the \emph{prime factor decomposition}, \emph{PFD} for short, w.r.t.\ 
the strong product of digraphs. The key idea for the prime
factorization of a strong product digraph $G=H\boxtimes K$ is the same
as for undirected graphs: We define the Cartesian skeleton $\skel(G)$
of $G$. The Cartesian
skeleton $\skel(G)$ is decomposed with respect to the Cartesian
product of digraphs. Afterwards, one determines the prime factors of
$G$ w.r.t.\  the strong product, 
using the information of the PFD of $\skel (G)$. This approach can
easily be extended if $G$ is not $S$-thin. In this contribution, 
we introduce the notion of the Cartesian skeleton of directed graphs 
 and show that it satisfies $\skel(H \boxtimes K) =
\skel(H)\Box \skel(K)$ for so-called ``$S$-thin'' digraphs. We prove
that $\skel(G)$ is connected whenever $G$ is connected and
provide new, fast and transparent algorithms for its construction.
Furthermore, we present the first polynomial-time algorithm for the computation of
the PFD w.r.t.\  the strong product of arbitrary connected digraphs. 


\section{Preliminaries}

\subsection{Basic Notation} 
A \emph{digraph} $G=(V,E)$ is a tupel consisting of a set of vertices
$V(G)=V$ and a set of ordered pairs $xy\in E(G)=E$, called (directed) edges or
arcs. In the sequel we consider only simple digraphs with finite vertex and
edge set. 
It is possible that both, $xy$ and $yx$ are contained in $E$. However, 
we only consider digraphs without loops, i.e., $xx\notin E$ for all $x\in V$. 
An \emph{undirected graph} $G=(V,E)$ is a tupel consisting of a set of vertices
$V(G)=V$ and a set of unordered pairs $\{x,y\}\in E(G)=E$. 
The \emph{underlying undirected} graph of a \emph{digraph} $G=(V,E)$ is the 
graph $U(G)=(V,F)$ with edge set
$F=\{\{x,y\}\mid xy\in E  \textrm{ or } yx\in E\}$.
A digraph $H$ is a \emph{subgraph} of a digraph $G$, in symbols
$H\subseteq G$, if $V(H)\subseteq V(G)$ and $E(H)\subseteq E(G)$. 
If in addition $V(H)=V(G)$, we call $H$ a \emph{spanning} subgraph
of $G$. If $H\subseteq G$ and all pairs of adjacent vertices in $G$  
are also adjacent in $H$ then $H$ is called an  \emph{induced} subgraph.
 The digraph $K_n=(V,E)$ with $|V|=n$ and $E=V\times V \setminus \{ (x,x) \mid x\in V\}$ is called a
\emph{complete graph}. 

A map $\gamma:V(H)\rightarrow V(G)$ 
such that $xy \in E(H)$ implies $\gamma(x)\gamma(y) \in E(G)$
for all $x,y \in V(G)$ is a \emph{homomorphism}.
We call two digraphs $G$ and $H$ \emph{isomorphic}, and write
$G\simeq H$, if there exists a bijective homomorphism $\gamma$ whose
inverse function is also a homomorphism. Such a map $\gamma$ is called an
\emph{isomorphism}.

Let $G=(V,E)$ be a digraph. The (closed) \emph{$N^+$-neighborhood} or
\emph{out-neighborhood} $N^+[v]$ of a vertex $v\in V$ is defined as
$N^+[v]=\{x\mid vx\in E\}\cup \{ v \}$. Analogously, the \emph{$N^-$-neighborhood}
or \emph{in-neighborhood} $N^-[v]$ of a vertex $v\in V$ is defined as
$N^-[v]=\{x\mid xv\in E\}\cup \{ v \}$. 
If there is a risk of confusion
we will write $N^+_G$, resp.,  $N^-_G$ 
to indicate that the respective neighborhoods are taken w.r.t.\  $G$. 
The \emph{maximum degree} $\Delta$ of a digraph $G=(V,E)$ is defined by
$\max_{v\in V} (|N^+[v]\setminus\{v\}|+|N^-[v]\setminus\{v\}|)$.

A digraph $G=(V,E)$ is \emph{weakly connected}, or \emph{connected} for short, 
if for every pair $x,y\in
V$ there exists a sequence $w=(x_0,\ldots ,x_n)$, called 
\emph{walk (connecting $x$ and $y$)} or just
\emph{xy-walk}, with $x=x_0$, $y=x_n$ such that $x_ix_{i+1}\in E
\textrm{\ or\ } x_{i+1}x_{i}\in E \textrm{\ for all\ } i \in \{
0,\ldots n-1\}.$ In other words, we call a digraph connected whenever
its underlying undirected graph is connected.

\subsection{The Cartesian and Strong Product}

The vertex set of the \emph{strong product} $G_1\boxtimes G_2$  of
two digraphs $G_1$ and $G_2$ is defined as 
$V(G_1)\times V(G_2) = \{(v_1,v_2)\mid v_1\in V(G_1), v_2\in V(G_2)\},$
Two vertices $(x_1,x_2)$, $(y_1,y_2)$ are adjacent
in $G_1\boxtimes G_2$ if one of the following conditions is
satisfied:
\begin{itemize} 
 \item[(i)] $x_1y_1\in E(G_1)$ and $x_2=y_2$,\vspace{-0.1in} 
 \item[(ii)]$x_2y_2\in E(G_2)$ and $x_1 = y_1$,\vspace{-0.1in} 
 \item[(iii)] $x_1y_1\in E(G_1)$ and $x_2y_2\in E(G_2)$.
\end{itemize}
The \emph{Cartesian product} $G_1  \Box G_2$ has the same vertex
set as $G_1\boxtimes G_2$, but vertices are only adjacent if they
satisfy (i) or (ii). 
Consequently, the edges of a strong product
that satisfy (i) or (ii) are called \emph{Cartesian}, the others
\emph{non-Cartesian}. 

The one-vertex complete graph $K_1$ 
serves as a unit for both products, as $K_1 \Box G = G$ and
$K_1 \boxtimes G = G$ for all graphs $G$.  
It is well-known that both products are associative
and commutative, see \cite{Hammack:2011a}.
Hence, a vertex $x$ of the strong product $\boxtimes_{i=1}^n G_i$ 
is properly ``coordinatized'' by the
vector $(x_1,\dots,x_n)$ whose entries are the vertices
$x_i$ of its factor graphs $G_i$.  
Therefore, the endpoints of
a Cartesian edge in a strong product
 differ in exactly one coordinate.

The Cartesian product and the strong product of digraphs is connected 
if and only if each of its factors is connected \cite{Hammack:2011a}.

In the product $\boxtimes_{i=1}^n G_i$, 
a \emph{$G_j$-layer} through vertex
$x$ with coordinates  $(x_1,\dots,x_n)$ is the induced subgraph $G_j^x$ in $G$
with vertex set
$\{(x_1,\dots x_{j-1},v,x_{j+1},\dots,x_n) \in V(G)\mid v \in
V(G_j)\}.$ Thus, $G_j^x$ is isomorphic to the factor $G_j$ for every $x\in V(G)$.  
For $y \in V(G_j^x)$ we have $G_j^x = G_j^y$, while $V(G_j^x) \cap V(G_j^z) = \emptyset$
if $z\notin V(G_j^x)$.

Finally, it is well-known that both products of connected digraphs satisfy 
the unique prime factorization property.

\begin{thm}[\cite{Fei86}]
Every finite simple connected digraph has a unique representation
as a Cartesian product of prime digraphs, up to isomorphism and order of the factors.
	\label{thm:uPFD-Cart}
\end{thm}

\begin{thm}[\cite{MCKE-71}]
Every finite simple connected digraph has a unique representation
as a strong product of prime digraphs, up to isomorphism and order of the factors.
\label{thm:uPFD-strong}
\end{thm}


In the sequel of this paper we will make frequent use of the fact that
for $G=G_1 \boxtimes G_2$ holds  
$N_G^+[(x,y)]=N_{G_1}^+[x] \times N_{G_2}^+[y]$ and 
$N_G^-[(x,y)]=N_{G_1}^-[x] \times N_{G_2}^-[y]$.



\subsection{The Relations $S^+$, $S^-$ and $S$ and Thinness}

It is important to notice that although the PFD w.r.t.\  the strong
product of connected digraphs is unique, the assignment of an edge being
Cartesian or non-Cartesian is not unique, in general. This is usually
possible if two vertices have the same out- and in-neighborhood.
Thus, an important issue in the context of strong products is whether
or not two vertices can be distinguished by their neighborhoods. This
is captured by the relation $S$ defined on the vertex set of $G$,
which was first introduced by D{\"o}rfler and Imrich \cite{DOIM-70}
for undirected graphs.

Let $G=(V,E)$ be a digraph. We define three equivalence relations on
$V$, based on respective neighborhoods. Two vertices $x,y\in V$
are in relation $S^+$, in symbols $x \sim_{S^+} y$, if $N_G^+[x] =
N_G^+[y]$. Analogously, $x,y\in V$ are in relation $S^-$ if $N_G^-[x]
= N_G^-[y]$. Two vertices $x,y\in V$ are in relation $S$ if $x
\sim_{S^+} y$ and $x \sim_{S^-} y$. Clearly, $S^+$, $S^-$ and $S$ are
equivalence relations. For a digraph $G$ let $S^+(v)=\{ u\in V(G) \,
|\, u \sim_{S^+} v\}$ denote the equivalence class of $S^+$ that
contains vertex $v$. Similarly,  $S^-(v)$ and $S(v)$ are defined.

We call a digraph $G=(V,E)$ \emph{S-thin} or \emph{thin} for short, 
if for all distinct vertices 
$x,y\in V$ holds $N_G^+[x]\neq N_G^+[y]$ or $N_G^-[x]\neq N_G^-[y]$.
Hence, a digraph is thin, if each equivalence class $S(v)$ of $S$ consists
of the single vertex $v\in V(G)$. 
In other words, $G$ is thin if all vertices can be distinguished by
their in- or out-neighborhoods.


The digraph $G/S$ is the usual \emph{quotient graph} 
with vertex set $V(G/S) = \{a\mid a \text{ is an equivalence class of }S \text{ in } G\}$
and $ab \in E(G/S)$ whenever $xy \in E(G)$ for some $x\in a$ and $y \in b$. 


In the following, we give several basic results concerning the relation $S$ and
quotients $G/S$ of digraphs $G$.
 
\begin{lem}
A digraph $G=G_1 \boxtimes G_2$ is thin if and only if $G_1$ and $G_2$ are thin.
\label{lem:thinFactors}
\end{lem}

\begin{proof}
Suppose that $G$ is not thin, and hence there are distinct vertices
$x=(x_1,x_2)\in V(G)$ and $y=(y_1,y_2)\in V(G)$ with
$N_G^+[(x_1,x_2)]= N_G^+[(y_1,y_2)]$ and  $N_G^-[(x_1,x_2)]= N_G^-[(y_1,y_2)]$.
This implies that $N_{G_1}^+[x_1] \times N_{G_2}^+[x_2]=N_{G_1}^+[y_1] \times
N_{G_2}^+[y_2]$. Hence,  $N_{G_1}^+[x_1]=N_{G_1}^+[y_1]$ and
$N_{G_2}^+[x_2]=N_{G_2}^+[y_2]$ and since $x\neq y$ we have
$x_1\neq y_1$ or $x_2\neq y_2$. Similar results hold for the 
$N^-$-neighborhoods.
Thus if $G$ is not thin, at least one of the factors is not thin.

On the other hand, if $G_1$ is not thin then
$N_{G_1}^+[x_1]=N_{G_1}^+[y_1]$ and $N_{G_1}^-[x_1]=N_{G_1}^-[y_1]$
for some $x_1\neq y_1$ and therefore $N_G^+[(x_1,z)]= N_G^+[(y_1,z)]$
and $N_G^-[(x_1,z)]= N_G^-[(y_1,z)]$ for all $z\in V(G_2)$.
\end{proof}

\begin{lem}
For any digraph $G=(V,E)$ the quotient graph $G/S$ is thin.
\label{lem:quotientthin}
\end{lem}

\begin{proof}
By definition of the relation $S$ for all $x,x'\in S(v)$
holds $N^+[x]=N^+[x']$ and $N^-[x]=N^-[x']$. Thus, 
there is an edge $xy\in E$, resp., $yx\in E$ for some $x\in S(v)$ 
if and only if for all $x'\in S(v)$ holds that 
$x'y\in E$, resp., $yx'\in E$.  
Thus, $ab \in E(G/S)$ if and only if 
for all $x\in a$ and $y\in b$ holds that $xy\in E$.

Assume $G/S$ is not thin. Then, there are distinct vertices 
$a,b\in V(G/S)$ with $S(a)=S(b)$ and hence, 
$N_{G/S}^+[a] = N_{G/S}^+[b] \textrm{ and } N_{G/S}^-[a]= N_{G/S}^-[b]$. 
Hence, $ac\in E(G/S)$ if and only if $bc\in E(G/S)$. 
By the preceding arguments, it holds that $ac\in E(G/S)$ if and only if 
for all $x\in a$ and $y\in c$ there is an edge $xy\in E$. 
Analogously, $bc\in E(G/S)$ if and only if 
for all $x'\in b$ and $y\in c$ there is an edge $x'y\in E$. 
Hence, $N_G^+[x] = N_G^+[x']$ for all $x\in a$ and $x'\in b$. 
By similar arguments one shows that $N_G^-[x] = N_G^-[x']$
for all $x\in a$ and $x'\in b$. 
But this implies that $a=S(x)=S(x')=b$, a contradiction. 
\end{proof}

\begin{lem}
Let $G$ be a digraph. Then
the subsets $S^+(v)$, $S^-(v)$ and  $S(v)$  induce complete
subgraphs for every vertex $v\in V(G)$.
\label{lem:eqcl-complete}
\end{lem}

\begin{proof}
If $S^+(v)=\{v\}$, then the assertion is clearly true.
Now, let $x,y \in S^+(v)$ be arbitrary.  By definition, 
$y\in N_G^+[y]$ and thus, $y\in N_G^+[x]$ and therefore, 
$xy\in E(G)$. Analogously, it hols that $x\in N_G^+[y]$ and thus,
$yx \in E(G)$.
Since this holds for all vertices contained in  
$S^+(v)$, they induce a complete 
graph $K_{|S^+(v)|}$.
By analogous arguments, the assertion is true for $S^-(v)$.
Since $S(v)=S^+(v)\cap S^-(v)$ for all $v\in V(G)$ and 
since $S^+(v)$ and $S^-(v)$ induce complete graphs,  
it follows that $S(v)$ induces a complete graph.
\end{proof}

\begin{lem}
For any digraphs $G$ and $H$ holds that 
$(G\boxtimes H)/S\simeq G/S\boxtimes H/S$
\label{lem:prodQuotients}
\end{lem}
\begin{proof}
Reasoning analogously as in the proof for undirected graphs in \cite[Lemma 7.2]{Hammack:2011a}, 
and by usage of Lemma \ref{lem:eqcl-complete} we obtain the desired result.
\end{proof}

\section{Dispensability and the Cartesian Skeleton}
\label{sec:CartSk}

A central tool for our PFD algorithms for connected digraphs $G$ 
is the \emph{Cartesian skeleton $\skel(G)$}.
The PFD of $\skel(G)$ w.r.t.\  the Cartesian product is utilized to
infer the prime factors w.r.t.\  the strong product of $G$.
This concept was first introduced for undirected graphs by Feigenbaum and
Sch{\"a}ffer in \cite{FESC-92} and later on improved by Hammack and Imrich,
see \cite{HAIM-09}. Following the illuminating approach of Hammack and Imrich, one
removes edges in $G$ that fulfill so-called dispensability conditions,
resulting in a subgraph $\skel(G)$ that is the desired Cartesian skeleton.
In this paper, we provide generalized 
dispensability conditions and thus, a general definition of 
the Cartesian skeleton of digraphs. 
For this purpose we first give the definitions of the so-called 
(weak) $N^+$-condition and  $N^-$-condition. Based on 
this, we will provide a general concept of dispensability
for digraphs, which in turn enables us to define the Cartesian skeleton $\skel(G)$. 
We prove that $\skel(G)$ is a connected
spanning subgraph, provided $G$ is connected. Moreover for $S$-thin
digraphs the Cartesian skeleton is uniquely determined and we obtain 
$\skel(H\boxtimes K) \simeq \skel(H)\Box \skel(K)$.

\begin{defi}
Let $G$ be a digraph and
 $xy\in E(G)$, $z\in V(G)$ be an arbitrary edge, resp, vertex of $G$.
We say $xy$ satisfies the 
\emph{$N^+$-condition with $z$} if one of the
following conditions is fulfilled:
\begin{enumerate}[(1$^+$)]
\item $N_G^+[x]\subset N_G^+[z] \subset N_G^+[y]$ 
\item $N_G^+[y]\subset N_G^+[z] \subset N_G^+[x]$ 
\item $N_G^+[x]\cap N_G^+[y] \subset N_G^+[x]\cap N_G^+[z]$ 
			 and $N_G^+[x]\cap N_G^+[y] \subset N_G^+[y]\cap N_G^+[z]$ 
\end{enumerate}

We say $xy$ satisfies the \emph{weak $N^+$-condition with $z$}, if the
following condition is fulfilled:
\begin{enumerate}
\item[] $N_G^+[x]\cap N_G^+[y] \subseteq N_G^+[x]\cap N_G^+[z]$ 
			  and $N_G^+[x]\cap N_G^+[y] \subseteq N_G^+[y]\cap N_G^+[z]$
\end{enumerate}
Analogously, by replacing ``$N_G^+$'' by ``$N_G^-$'' we get Conditions (1$^-$),(2$^-$),(3$^-$), 
for the definition of the \emph{$N^-$-condition with $z$}, respectively, for the definition
of the \emph{weak $N^-$-condition with $z$}.
\end{defi}

\begin{defi}
Let $G$ be a digraph. An edge $xy\in E(G)$ is
\emph{dispensable} if at least one of the following conditions is satisfied:
\begin{enumerate}[(D1)]
	\item There exists a vertex $z\in V(G)$ such that $xy$ satisfies the $N^+$- and
	      	$N^-$-condition with $z$.
	\item There are vertices $z_1,z_2 \in V(G)$ such that both conditions holds:
	\begin{enumerate}[(a)]
		\item $xy$ satisfies $(3^+)$ of the $N^+$-condition with $z_1$ and the 
					 weak $N^-$-condition with $z_1$. 
		\item $xy$ satisfies  $(3^-)$ of the $N^-$-condition with $z_2$ and the 
					 weak $N^+$-condition with $z_2$. 
	\end{enumerate} 
	\item There exists a vertex $z\in V(G)$ such that $xy$ satisfies
	      the $N^+$-condition with $z$ and at least one of the following
	      holds: $N^-[x]=N^-[z]$ or $N^-[y]=N^-[z]$.
	\item There exists a vertex $z\in V(G)$ such that $xy$ satisfies
	      the $N^-$-condition with $z$ and at least one of the following
	      holds: $N^+[x]=N^+[z]$ or $N^+[y]=N^+[z]$.
	\item There are distinct vertices $z_1,z_2 \in V(G)$, both distinct from
	      $x$ and $y$, such that $N^+[x]=N^+[z_1]$, $N^-[x]=N^-[z_2]$,
	      $N^-[z_1]=N^-[y]$ and $N^+[z_2]=N^+[y]$.
\end{enumerate}
All other edges in $E(G)$ are \emph{non-dispensable}. 
\end{defi}

Note, if one considers undirected graphs $G=(V,E)$ as graphs for which
$N^+[v]=N^-[v]$	for all $v\in V$, then
none of the Conditions $(D2)$-$(D4)$ can be fulfilled for $G$.
Moreover if this undirected graph is thin, then Condition $(D5)$ 
cannot be satisfied. 
In other words, the definition of dispensability 
reduces to $(D1)$ and thus, 
coincides with that for undirected graphs given by Hammack 
and Imrich \cite{HAIM-09}.

\begin{rem}
	\label{rem:edge}
	Let $G = (V,E)$ be a digraph and assume the edge $xy\in E$
	is dispensable by one of the Conditions $(D1)$, $(D3)$
	$(D4)$ with	some vertex $z\in V$ or $(D2)$, $(D5)$
	with some  $z_1,z_2\in V$. It is now an easy task to verify that
	$z\in N^+[x]\cup N^-[x]$ and	$z\in N^+[y]\cup N^-[y]$.
	The same is true for $z_1$ and $z_2$.
\end{rem}


We are now in the position to define the Cartesian skeleton
of digraphs.

\begin{defi}
The \emph{Cartesian skeleton} of a digraph $G$ is the digraph $\skel(G)$
that is obtained from $G$ by removing all dispensable edges. 
More precise, the Cartesian skeleton $\skel(G)$ has vertex set $V(G)$ and
edge set $E(\skel(G))=E(G)\setminus D(G)$, where $D(G)$ denotes
the set of dispensable edges in $G$.
\end{defi}

\begin{figure}[tbp]
  \centering
	\includegraphics[bb=26 321 444 648, scale=0.6]{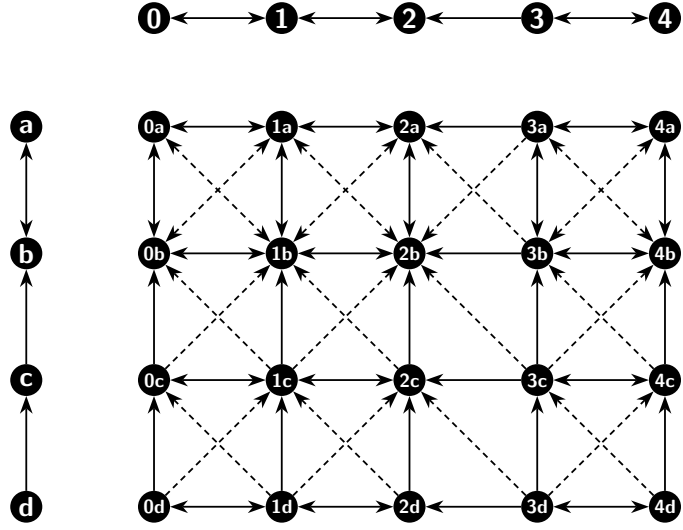}
	\caption{Shown is the strong product of two thin digraphs $G_1$ and $G_2$. 
						The dashed edges are dispensable and thus, $\skel(G_1\boxtimes G_2)
						=\skel(G_1)\Box \skel(G_2)$ is the subgraph that contains
						all non-dashed edges. By way of example, the edge 
						$(0d)(1c)$ satisfies $(D1)$ with $z=(1d)$; 
						the edge $(2d)(1c)$ satisfies $(D2)$ with $z_1=(1d), z_2=(2c)$;
						the edge $(3c)(4b)$ satisfies $(D3)$ with $z=(3b)$; 
						the edge $(3a)(2b)$ satisfies $(D4)$ with $z=(3b)$;
						and the edge $(4a)(3b)$ satisfies $(D5)$ with $z_1=(4b), z_2=(3a)$.}
	\label{fig:strong}
\end{figure}

In the following, we will show that non-Cartesian edges are dispensable and
moreover that $\skel(H \boxtimes K) = \skel(H)\,\Box \, \skel(K)$, whenever
$H$ and $K$ are thin graphs. 

\begin{lem}
Let $G=H\boxtimes K$ be a thin digraph. Then every non-Cartesian edge is dispensable
and thus, every edge of  $\skel(G)$ is Cartesian w.r.t.\   this factorization. 
\label{lem:nonCdisp}
\end{lem}

\begin{proof}
Suppose that the edge $(h,k)(h',k')\in E(G)$ is non-Cartesian. 
We have to examine several cases.

Assume $N_H^+[h]\neq N_H^+[h']$ and $N_K^+[k]\neq N_K^+[k']$. Then
\begin{align}
N^+_G[(h,k)]\cap N^+_G [(h',k')] &= 
	(N_H^+[h] \cap N_H^+ [h']) \times (N_K^+[k] \cap N_K^+ [k']) \notag\\
& \subseteq N_H^+[h]  \times (N_K^+[k] \cap N_K^+ [k']) \notag\\
& = N^+_G[(h,k)]\cap N^+_G [(h,k')] \label{incl1}
\end{align}
\begin{align}
N^+_G[(h,k)]\cap N^+_G [(h',k')] &= (N_H^+[h] \cap N_H^+ [h']) \times (N_K^+[k] \cap N_K^+ [k']) \notag\\
& \subseteq (N_H^+[h] \cap N_H^+ [h']) \times N_K^+ [k'] \notag\\
& = N^+_G[(h,k')]\cap N^+_G [(h',k')] \label{incl2}
\end{align}

Interchanging the roles of $h$ and $k$ with $h'$ and $k'$ gives us by similar arguments:
\begin{align}
N^+_G[(h',k')]\cap N^+_G [(h,k)] &\subseteq N^+_G[(h',k')]\cap N^+_G [(h',k)] \label{incl3} \text{ and }\\
N^+_G[(h',k')]\cap N^+_G [(h,k)] &\subseteq N^+_G[(h',k)]\cap N^+_G [(h,k)]. \label{incl4}
\end{align}

Notice that $N^+_G[(h,k)]\cap N^+_G [(h',k')]\neq \emptyset$, since 
 $(h,k)(h',k')\in E(G)$ implies that $(h',k')\in N^+_G[(h,k)]\cap N^+_G [(h',k')]$.
The following four cases can occur:

\begin{enumerate}
\item All inclusions in Eq.\ \eqref{incl1}\ -\ \eqref{incl4} are inequalities, thus $(h,k)(h',k')$ satisfies $(3^+)$
	    of the $N^+$-condition  with $z$ by choosing $z=(h,k')$ or
      $z=(h',k)$.
\item Only the first two inclusions (Eq.\ \eqref{incl1}\ -\ \eqref{incl2}) are inequalities, thus $(h,k)(h',k')$
      satisfies $(3^+)$ of the $N^+$-condition with $z=(h,k')$ and 
			the weak $N^+$-condition with $z=(h',k)$.
\item Symmetrically, if only the last two inclusions (Eq.\ \eqref{incl3}\ -\ \eqref{incl4}) are inequalities, then
      $(h,k)(h',k')$ satisfies $(3^+)$ of the $N^+$-condition with
      $z=(h',k)$ and the weak $N^+$-condition with $z=(h,k')$.	
\item At least one of the first two and one of last two inclusions are
      equality. 
			From the first two formulas we get $N_H^+[h]\cap N_H^+[h'] = N_H^+[h]$ 
			or $N_K^+[k]\cap N_K^+[k'] = N_K^+ [k']$. Due to the
			assumption $N_K^+[h]\neq N_K^+[h']$ and $N_K^+[k]\neq N_K^+[k']$
			this implies
			$$ N_H^+[h] \subset N_H^+[h'] \textrm{  or  } N_K^+[k'] \subset N_K^+[k].$$
			Similarly we get from the last two formulas
			$$ N_H^+[h'] \subset N_H^+[h] \textrm{  or  }   N_K^+[k] \subset N_K^+[k'].$$
			This implies we have
	$$ N_H^+[h] \subset N_H^+[h'] \textrm{  and  }   N_K^+[k] \subset N_K^+[k'] 
	\textrm{\ or\ } 
  	N_H^+[h'] \subset N_H^+[h] \textrm{  and  }   N_K^+[k'] \subset N_K^+[k]$$
	and thus 
		$$N^+_G[(h,k)] \subset N^+_G[(h,k')] \subset N^+_G[(h',k')]
		\textrm{ and } N^+_G[(h,k)] \subset N^+_G[(h',k)] \subset
		N^+_G[(h',k')]$$
		or
		$$N^+_G[(h',k')] \subset N^+_G[(h,k')] \subset N^+_G[(h,k)] \textrm{
		and } N^+_G[(h',k')] \subset N^+_G[(h',k)] \subset N^+_G[(h,k)].$$
		Therefore, also in this case $(h,k)(h',k')$ satisfies the
		$N^+$-condition with $z=(h,k')$ and with $z=(h',k)$.
\end{enumerate}

So far we treated the $N^+$-neighborhoods under the assumption that
$N_H^+[h]\neq N_H^+[h']$ and $N_K^+[k]\neq N_K^+[k']$. For the $N^-$-
neighborhoods the situation can be treated analogously, if we assume that
$N_H^-[h]\neq N_H^-[h']$ and $N_K^-[k]\neq N_K^-[k']$. Then, we obtain
the same latter four cases just by replacing $N_H^+$ and $N_K^+$, by
$N_H^-$ and $N_K^-$, respectively. Now, it is easy to verify that 
every combination of the Cases $1.$\ -\ $4.$ for $N^+$- and $N^-$-neighborhoods
leads to one of the conditions $(D1)$ or $(D2)$. 

Assume that $N_H^+[h]= N_H^+[h']$ and $N_K^-[k] = N_K^-[k']$. 
Then Condition $(D5)$ holds for the edge $(h,k)(h',k')$ with $z_1=(h',k)$ and
$z_2=(h,k')$. Analogous arguments show that Condition $(D5)$ is satisfied, if 
$N_H^-[h]= N_H^-[h']$ and $N_K^+[k] = N_K^+[k']$.

Finally, assume that $N_H^+[h]= N_H^+[h']$ and $N_K^-[k] \neq N_K^-[k']$. 
By thinness it must hold $N_H^-[h]\neq N_H^-[h']$. 
Thus, we have the Cases $1.$\ -\ $4.$ for $N^-$-neighborhoods.
In particular, for all four cases we can infer that the edge $(h,k)(h',k')$ 
satisfies the $N^-$-condition with vertex $(h,k')$ or $(h',k)$. 
Hence, Condition $(D4)$ is satisfied 
since $N_G^+[(h,k)]= N_G^+[(h',k)]$ and $N_G^+[(h,k')]= N_G^+[(h',k')]$.
If  $N_H^+[h] \neq N_H^+[h']$ and $N_K^-[k] = N_K^-[k']$ then 
we obtain by similar arguments, that $(D3)$ is satisfied. 

Hence, in all cases we can observe that non-Cartesian edges fulfill
one of the Condition $(D1)-(D5)$ and are thus, dispensable. 
\end{proof}

\begin{lem}
If $H$, $K$ are thin digraphs,
then $\skel(H \boxtimes K) \subseteq \skel(H)\,\Box \, \skel(K)$.
\label{lem:skel1}
\end{lem}

\begin{proof}
In the following, we will denote in some cases for simplicity the
product $H \boxtimes K$ by $G$. By Lemma \ref{lem:nonCdisp}, the
subgraph $\skel(H \boxtimes K)$ contains Cartesian edges only. Hence, by
commutativity of the Cartesian product, it remains to show that for
every non-dispensable Cartesian edge $(h,k)(h',k)$ contained in $\skel(H
\boxtimes K)$, there is an edge $hh' \in \skel(H)$ and thus $(h,k)(h',k)$
is also contained in $\skel(H)\,\Box \, \skel(K)$.

By contraposition, assume that $hh'$ is dispensable in $H$, that is, one of
the Conditions $(D1)$-$(D5)$ is fulfilled. 

Assume $(D1)$ holds for $hh'$ with some $z\in V(H)$. Then one of the
following conditions holds 
$(1^+)$ $N_H^+[h]\subset N_H^+[z] \subset N_H^+[h']$, $(2^+)$
$N_H^+[h']\subset N_H^+[z] \subset N_H^+[h]$ or 
$(3^+)$  $N_H^+[h]\cap N_H^+[h'] \subset N_H^+[h]\cap N_H^+[z]$ and
$N_H^+[h]\cap N_H^+[h'] \subset N_H^+[h']\cap N_H^+[z]$.
If we multiply every neighborhood in the inclusions with $N_K^+[k]$ we get a 
$N^+$-condition for $(h,k)(h',k)$ with $(z,k)$. Analogously, if $hh'$
satisfies the $N^-$-condition with $z\in V(H)$, then $(h,k)(h',k)$
satisfies $N^-$-condition with $(z,k)$. Thus Condition $(D1)$ for
$hh'$ implies $(D1)$ for $(h,k)(h',k)$.

Assume $(D2)$ holds for $hh'$. Hence there are vertices $z_1,z_2\in
V(H)$ s.t. $hh'$ satisfies $(3^+)$ of the $N^+$-condition with $z_1$
and the weak $N^-$-condition with $z_1$, as well as, the $(3^-)$ of
the $N^-$-condition with $z_2$ and the weak $N^+$-condition with
$z_2$. As argued before, the edge $(h,k)(h',k)$ satisfies $(3^+)$ of
the $N^+$-condition with $(z_1,k)$ and $(3^-)$ of the $N^-$-condition
with $(z_2,k)$. For $hh'$ and the weak $N^-$-condition holds
$N_H^-[h]\cap N_H^-[h'] \subseteq N_H^-[h]\cap N_H^-[z_1]$ and  
$N_H^-[h]\cap N_H^-[h'] \subseteq N_H^-[h']\cap N_H^-[z_1]$.
Again, if we multiply every inclusion with $N^-[k]$ we can infer that 
 $$N_G^-[(h,k)]\cap N_G^-[(h',k)] \subseteq N_G^-[(h,k)]\cap N_G^-[(z_1,k)]$$
  and
  $$N_G^-[(h,k)]\cap N_G^-[(h',k)] \subseteq N_G^-[(h',k)]\cap N_G^-[(z_1,k)].$$
Thus Item $(a)$ of Condition $(D2)$ is satisfied for $(h,k)(h',k)$
with $(z_1,k)$. By analogous arguments, we derive that Item $(b)$ of
Condition $(D2)$ is satisfied for $(h,k)(h',k)$ with $(z_2,k)$. Hence,
Condition $(D2)$ for $hh'$ implies that $(D2)$ holds for
$(h,k)(h',k)$.

For Condition $(D3)$, resp., $(D4)$ we can infer by the preceding
arguments, that the $N^+$-condition, resp., $N^-$-condition for
$(h,k)(h',k)$ with $(z,k)$ is fulfilled, whenever these conditions are
satisfied for $hh'$ with $z$. Now, $N_H^-[h]=N_H^-[z]$ or
$N_H^-[h']=N_H^-[z]$ implies $N_G^-[(h,k)]=N_G^-[(z,k)]$ or
$N_G^-[(h',k)]=N_G^-[(z,k)]$ and similarly this holds for $N^+$-
neighborhoods. Hence $(D3)$, resp., $(D4)$ are fulfilled for the edge
$(h,k)(h',k)$. 

Finally, consider Condition $(D5)$. Assume there are distinct vertices
$z_1,z_2 \in V(G)$ such that $N_H^+[h]=N_H^+[z_1]$,
$N_H^-[h]=N_H^-[z_2]$, $N_H^-[z_1]=N_H^-[h']$ and
$N_H^+[z_2]=N_H^+[h']$. This implies that
$N_G^+[(h,k)]=N_G^+[(z_1,k)]$, $N_G^-[(h,k)]=N_G^-[(z_2,k)]$,
$N_G^-[(z_1,k)]=N_G^-[(h',k)]$ and $N_G^+[(z_2,k)]=N_G^+[(h',k)]$ and
therefore, Condition $(D5)$ is fulfilled for the edge $(h,k)(h',k)$. 

To summarize, if $hh'$ is dispensable then $(h,k)(h',k)$ is
dispensable and hence, $\skel(H \boxtimes K) \subseteq \skel(H)\,\Box \,
\skel(K)$.
\end{proof}

\begin{prop}
If $H$, $K$ are thin graphs, then $\skel(H \boxtimes K) = \skel(H)\, \square
\, \skel(K)$.
\label{prop:skelProd}
\end{prop}

\begin{proof}
By Lemma \ref{lem:skel1}, it remains to prove that $\skel(H)\, \square \,
\skel(K)\subseteq \skel(H \boxtimes K)$. Moreover, by commutativity of the
products, we must only show that for every edge $(h,k)(h',k)\in
E(\skel(H)\, \square \, \skel(K))$ holds that $(h,k)(h',k)$ is not dispensable
in $H \boxtimes K$. 

For contraposition, assume $(h,k)(h',k)$ is dispensable in $H
\boxtimes K$. We will prove that then $hh'$ is dispensable in $H$.
In the following, we will denote in some cases for simplicity the
product $H \boxtimes K$ by $G$.

Let us assume that Condition $(D1)$ holds for $(h,k)(h',k)$ with
$z=(z',z'')$. If Condition $(1^+)$ is fulfilled then
$N_G^+[(h,k)]\subset N_G^+[(z',z'')] \subset N_G^+[(h',k)]$ and we get
$$N_H^+[h]\times N_K^+[k]\subset N_H^+[z']\times N_K^+[z''] \subset
N_H^+[h']\times N_K^+[k].$$ The latter implies that 
$N_K^+[z'']=N_K^+[k]$, which causes $N_H^+[h]\subset N_H^+[z'] \subset
N_H^+[h']$ and hence, $(1^+)$ is fulfilled in $H$ for $hh'$ with $z'$.
If Condition $(2^+)$ is fulfilled, then analogous arguments show that
$N_H^+[h']\subset N_H^+[z'] \subset N_H^+[h]$.
assume now that Condition $(3^+)$  holds: 
$N_G^+[(h,k)]\cap N_G^+[(h',k)] \subset N_G^+[(h,k)]\cap N_G^+[(z',z'')]$ and 
$N_G^+[(h,k)]\cap N_G^+[(h',k)] \subset N_G^+[(h',k)]\cap N_G^+[(z',z'')]$. 
Therefore, 
$$(N_H^+[h]\cap N_H^+[h'])\times N_K^+[k] \subset (N_H^+[h]\cap N_H^+[z']) \times (N_K^+[k]\cap N_K^+[z''])\ $$ and
$$(N_H^+[h]\cap N_H^+[h'])\times N_K^+[k] \subset (N_H^+[h']\cap N_H^+[z']) \times (N_K^+[k]\cap N_K^+[z'']).$$
Since $hh'\in E(H)$, we can conclude that $N_H^+[h]\cap N_H^+[h']\neq \emptyset$. 
Hence, the latter implies that $N_K^+[k]\subseteq N_K^+[k]\cap N_K^+[z'']$ and
thus, $N_K^+[k]=N_K^+[k]\cap N_K^+[z'']$. Then it must holds that
$N_H^+[h]\cap N_H^+[h'] \subset N_H^+[h]\cap N_H^+[z']$ and
$N_H^+[h]\cap N_H^+[h'] \subset N_H^+[h']\cap N_H^+[z']$, which yields
$(3^+)$ for $hh'$ with $z'$. Similarly all $N^-$-conditions can be
transferred from $(h,k)(h',k)$ with $(z',z'')$ to $hh'$ with $z'$.
Hence whenever Condition $(D1)$ if fulfilled for $(h,k)(h',k)$ with
$z=(z',z'')$ then $(D1)$ holds for $hh'$ with $z'$, as well. 

Now, assume that Condition $(D2)$ holds for $(h,k)(h',k)$ with
$z_1=(z_1',z_1'')$ and $z_2=(z_2',z_2'')$. By the above arguments it
is clear that $(3^+)$ is fulfilled for $hh'$ with $z_1'$ and $(3^-)$
is fulfilled for $hh'$ with $z_2'$. Consider the weak $N^-$-condition
for $(h,k)(h',k)$ with $z_1$: 
$$N_G^-[(h,k)]\cap N_G^-[(h',k)] \subseteq N_G^-[(h,k)]\cap N_G^-[(z_1',z_1'')]$$ and 
$$N_G^-[(h,k)]\cap N_G^-[(h',k)] \subseteq N_G^-[(h',k)]\cap N_G^-[(z_1',z_1'')].$$
Obviously this implies $N_H^-[h]\cap N_H^-[h'] \subseteq N_H^-[h]\cap
N_H^-[z_1']$ and $N_H^-[h]\cap N_H^-[h'] \subseteq N_H^-[h']\cap
N_H^-[z_1']$. Therefore, the weak $N^-$-condition holds for $hh'$ with
$z_1'$. By analogous arguments, we obtain that also the weak
$N^+$-condition is fulfilled for $hh'$ with $z_2'$. Hence, Condition
$(D2)$ holds for $hh'$ with $z_1'$ and $z_2'$.

If Condition $(D3)$ is fulfilled for  $(h,k)(h',k)$ with
$z=(z',z'')$ then by the above arguments, the $N^-$-condition holds
for $hh'$ with $z'$. Moreover, it holds $N_G^-[(h,k)]=N_G^-[(z',z'')]$
or $N_G^-[(h',k)]=N_G^-[(z',z'')]$, but this is only possible if
$N_H^-[h]=N_H^-[z']$ or $N_H^-[h']=N_H^-[z']$. Hence, $(D3)$ holds for
$hh'$ with $z'$. By analogous arguments we can infer that Condition
$(D4)$ holds for $hh'$ with $z'$ whenever $(D4)$ holds for
$(h,k)(h',k)$ with $z=(z',z'')$.

It remains to check Condition $(D5)$. Let $(z_1',z_1'')$ and
$(z_2',z_2'')$ be two distinct vertices such that
$N_G^+[(h,k)]=N_G^+[(z_1',z_1'')]$,
$N_G^-[(h,k)]=N_G^-[(z_2',z_2'')]$,
$N_G^-[(z_1',z_1'')]=N_G^-[(h',k)]$ and
$N_G^+[(z_2',z_2'')]=N_G^+[(h',k)]$. Again, this is only possible if
$N_H^+[h]=N_H^+[z_1']$, $N_H^-[h]=N_H^-[z_2']$,
$N_H^-[z_1']=N_H^-[h']$ and $N_H^+[z_2']=N_H^+[h']$. Clearly, since
$(h,k)(h',k)\in E(G)$ the vertices $h$ and $h'$ are distinct. However,
we must also verify that $z_1'\neq z_2'$ and
$z_1',z_2'\notin\{h,h'\}$.

Assume $z_1'=h$. Since by assumption,
$N_G^-[(z_1',z_1'')]=N_G^-[(h',k)]$ it must hold
$N_K^-[z_1'']=N_K^-[k]$. Then we can infer that $N_H^-[h]\times
N_K^-[k]= N_H^-[h]\times N_K^-[z_1''] = N_H^-[z_1']\times
N_K^-[z_1'']$ and thus, $N_G^-[(h,k)]=N_G^-[(z_1',z_1'')]$. However,
this contradicts the fact that $G$ is thin, since we assumed that
$N_G^+[(h,k)]=N_G^+[(z_1',z_1'')]$. Using analogous arguments one 
shows that $z_1',z_2'\notin\{h,h'\}$.

Finally, assume that $z_1'=z_2'$. First, note that
$N_G^-[(z_1',z_1'')]=N_G^-[(h',k)]$ implies that
$N_K^-[z_1'']=N_K^-[k]$. Second, $N_G^-[(h,k)]=N_G^-[(z_2',z_2'')]$
implies that $N_H^-[h]=N_H^-[z_2']$ and thus, $N_H^-[h]=N_H^-[z_1']$.
Therefore, by the same arguments as before, we obtain that
$N_G^-[(h,k)]=N_G^-[(z_1',z_1'')]$, which contradicts that $G$ is
thin, since by assumption $N_G^+[(h,k)]=N_G^+[(z_1',z_1'')]$. Hence,
Condition $(D5)$ is fulfilled for $hh'$ with $z_1'$ and $z_2'$. 

To summarize, dispensability of $(h,k)(h',k)$ in $H \boxtimes K$
implies dispensability of $hh'$ in $H$. By commutativity of the
products, we can conclude that $\skel(H)\, \square \, \skel(K)\subseteq \skel(H
\boxtimes K)$, that together with Lemma \ref{lem:skel1} implies $\skel(H
\boxtimes K)=\skel(H)\, \square \,\skel(K)$
\end{proof}

In the following, we will show that the Cartesian skeleton $\skel(G)$ of a connected thin digraph
$G$ is connected.  


\begin{lem}
Let $G=(V,E)$ be a thin connected digraph  and let $S^+(v)$ and $S^-(v)$ be
the corresponding $S^+$- and $S^-$-classes containing vertex $v\in V$.
Then all vertices contained in $S^+(v)$ lie  in the
same connected component of $\skel(G)$, i.e., there is always 
a walk
consisting of non-dispensable edges only, that  connects 
all vertices $x,y\in S^+(v)$ . The same is true for all vertices 
contained in  $S^-(v)$.
\label{lem:cC}
\end{lem}

\begin{proof}
If $|S^+(v)|=1$ there is nothing to show. Thus, assume $x,y\in
S^+(v)$. By Lemma \ref{lem:eqcl-complete} there is an edge $xy\in E(G)$.
Assume that this edge $xy$ is dispensable. Since $N^+[x]=N^+[y]$, 
none of the Conditions $(D1)$, $(D2)$, and $(D3)$ can be satisfied
for the edge $xy$. Moreover, $(D5)$ can not hold, since otherwise we would have
$N^+[x]=N^+[z_1]=N^+[y]=N^+[z_2]$ and $N^-[x]=N^-[z_2]$ and thus, $G$ would not
be thin. 
Therefore, if $xy$ is dispensable, then Condition $(D4)$ must hold. 
Thus, there is a vertex $z$ such that one of the $N^-$-conditions
$(1^-)$, $(2^-)$ or  $(3^-)$  with $z$ holds for $xy$ 
and $N^+[x]=N^+[z]$ or $N^+[z]=N^+[y]$. Since $N^+[x]=N^+[y]$, we can conclude that
$z$ must be contained in $S^+(v)$.


First assume that Condition $(1^-)$ for $xy$ with $z$ is satisfied and
therefore in particular, $N^-[x]\subset N^-[y]$. Consider the maximal
chain $N^-[x]\subset N^-[z_1]\subset \ldots \subset N^-[z_{k-1}]
\subset N^-[y]$ of neighborhoods between $N^-[x]$ and $N^-[y]$ ordered
by proper inclusions, where $z_i\in S^+(v)$ for all $i\in
\{1,\dots,k-1\}$. To simplify the notation let $x=z_0$ and $y=z_k$.
Lemma \ref{lem:eqcl-complete} implies that $z_iz_{i+1} \in E(G)$ for
all $i\in\{0,\dots,k-1\}$. We show that the edges $z_iz_{i+1}$ for all
$i\in \{0,\dots,k-1\}$ are non-dispensable. By the preceding
arguments, such an edge $z_iz_{i+1}$ can only be dispensable if
Condition $(D4)$ is satisfied, and thus in particular, if there exists
a vertex $z'\in S^+(v)$ such that the $N^-$-condition for $z_iz_{i+1}$
holds with $z'$. Since $N^-[z_{i}]\subset N^-[z_{i+1}]$ we can
conclude that Condition $(2^-)$ cannot be satisfied. Moreover,
$N^-[z_{i}] \cap N^-[z_{i+1}] \subset N^-[z_{i}] \cap N^-[z]$ is not
possible, and thus, Condition $(3^-)$ cannot be satisfied.
Furthermore, since we constructed a maximal chain of proper included
neighborhoods, $N^-[z_{i}] \subset N^-[z']\subset N^-[z_{i+1}]$ is not
possible and and therefore, Condition $(1^-)$ cannot be satisfied.
Hence, none of the $N^-$-conditions for the edges $xz_1,z_1z_2,\ldots
, z_ky$ can be satisfied, which yields a walk in $\skel(G)$ connecting
$x$ and $y$. Therefore, all $x,y\in S^+(v)$ with $N^-[x]\subset
N^-[y]$ lie in the same connected component of $\skel(G)$. By
analogous arguments one shows that $x$ and $y$ lie in the same
connected component of $\skel(G)$ if Condition $(2^-)$ and thus,
$N^-[y]\subset N^-[x]$ is satisfied. 

We summarize at this point: All $x,y\in S^+(v)$ where the edge $xy$ 
fulfill Condition $(1^-)$ or $(2^-)$ are in the same connected component of 
$\skel(G)$. Now assume for contradiction,  that there are 
vertices  $x,y \in S^+(v)$ (and hence, an edge $xy \in E(G)$)
that are in different connected components of $\skel(G)$. 
This is only possible if the edge $xy$ is dispensable by Condition 
$(3^-)$ and thus if $N^-[x] \cap N^-[y] \subset N^-[x] \cap N^-[z]$ and 
$N^-[x] \cap N^-[y] \subset N^-[y] \cap N^-[z]$. 
Define for arbitrary vertices $x,y \in S^+(v)$ the 
integer $k_{xy}=|N^-[x]\cap N^-[y]|$ and take among all  
$x,y \in S^+(v)$ that are in different connected components 
of $\skel(G)$ the 
ones that have largest value $k_{xy}$. Note,
$k_{xz},k_{yz}>k_{xy}$. Moreover, since $z\in S^+(v)$ and
we have taken $x,y\in S^+(v)$ that have largest integer $k_{xy}$ 
among all vertices that are in different connected components 
of $\skel(G)$, we can conclude that $x$ and $z$, as well as
$y$ and $z$ are in the same connected component in $\skel(G)$,  
a contradiction.
This completes the proof for the case $x,y\in S^+(v)$.

By analogous arguments one shows that the statement is true for 
$S^-(v)$.
\end{proof}

\begin{lem}
Let $G=(V,E)$ be a thin connected digraph and $x,y\in V$ with
$N^+[x]\subset N^+[y]$ or $N^-[x]\subset N^-[y]$. 
Then there is a walk in $\skel(G)$ connecting $x$ and $y$.
\label{lem:subConnected}
\end{lem}

\begin{proof}
Assume first that $N^+[x]\subset N^+[y]$. 
Note, one can always construct a maximal chain of vertices with $N^+[x]\subset
N^+[z_1]\subset \ldots \subset N^+[y]$ and connect walks inductively, 
whenever the statement is true.
Therefore, we can assume that $N^+[x]\subset N^+[y]$ with no $z\in
V$ such that $N^+[x]\subset N^+[z] \subset N^+[y]$. Clearly, 
$N^+[x]\subset N^+[y]$ implies $xy\in
E(G)$. Assume $xy$ is dispensable. Since by assumption there is no $z$
with $N^+[x]\subset N^+[z] \subset N^+[y]$ it follows that Condition
$(1^+)$ can not hold. Moreover, since $N^+[x]\subset N^+[y]$ Condition
$(2^+)$ can not hold. The latter also implies $N^+[x] \cap
N^+[y]=N^+[x]$ and thus Condition $(3^+)$ can not be fulfilled, since
$N^+[x]=N^+[x] \cap N^+[y] \subset N^+[x] \cap N^+[z]$ is not
possible. Thus $xy$ does not satisfy the $N^+$-condition and thus it 
cannot be dispensable by Conditions $(D1)$, $(D2)$ or $(D3)$. 

If it is dispensable by Condition $(D5)$, then there exists a vertex
$z_1$ with $N^+[x]=N^+[z_1]$ and $N^-[z_1]=N^-[y]$. Hence, $z_1\in
S^+(x)$ and $z_1\in S^-(y)$. By Lemma \ref{lem:cC} there is a
$x,z_1$-walk and $z_1,y$-walk and thus, a walk connecting $x$ and $y$
consisting of non-dispensable edges only. This together with Lemma
\ref{lem:cC} implies that, also all vertices $x'\in S^+(x)$ and $y'\in
S^+(y)$ are connected by a walk of non-dispensable edges. 

Assume now for contradiction that vertices $x$ and $y$ are in different 
connected components of $\skel(G)$. By Lemma \ref{lem:cC}, all vertices contained 
 $S^+(x)$ are in same connected component of $\skel(G)$. The same is
true for all vertices contained in $S^+(y)$. 
Hence if $x$ and $y$ are in different components then all vertices contained
 $S^+(x)$ must be in a different connected component of $\skel(G)$ than the 
vertices contained in $S^+(y)$. By the preceding arguments, 
this can only happen, when
all edges $x'y'\in E$ with $x'\in S^+(x)$ and $y'\in S^+(y)$ are dispensable 
by Condition $(D4)$. 
We examine now three cases: there are $x'\in S^+(x)$ and $y'\in S^+(y)$
with $(i)$ $N^-[x']\subset N^-[y']$, $(ii)$
$N^-[y']\subset N^-[x']$ or $(iii)$ none of the cases $(i),(ii)$ hold. 

\emph{Case (i) $N^-[x']\subset N^-[y']$:} W.l.o.g. assume that $x'$
and $y'$ are chosen, such that $|N^-[y']-N^-[x']|$ becomes minimal
among all such pairs $x' \in S^+(x)$ and $y' \in S^+(y)$. Since the
edge $x'y'$ must be dispensable by Condition $(D4)$, there is a vertex
$z\in V$ with $N^+[x']=N^+[z]$ or $N^+[y']=N^+[z]$ and $xy$ satisfies
the $N^-$-condition with $z$. Clearly, Condition $(2^-)$ with
$N^-[y']\subset N^-[z]\subset N^-[x']$ cannot be fulfilled. Moreover,
Condition $(3^-)$ cannot hold, since $N^-[x']\subset N^-[y']$ implies
that $N^-[x']=N^-[x'] \cap N^-[y']$ and thus $N^-[x'] \cap N^-[y']\subset
N^-[x'] \cap N^-[z]$ is not possible. Thus, assume $x'y'$ fulfills
Condition $(1^-)$ with $z$, then $N^-[x']\subset N^-[z]\subset
N^-[y']$. Since $N^+[x']=N^+[z]$ or $N^+[y']=N^+[z]$ we have that $z
\in S^+(x)$ or $z \in S^+(y)$. But then $|N^-[z]-N^-[x']|$ or
$|N^-[y']-N^-[z]|$ is smaller than $|N^-[y']-N^-[x']|$, a
contradiction. Hence, $x'y'$ is \emph{not} dispensable and thus
the vertices in $S^+(x)$ and $S^+(y)$ cannot be 
in different connected components of $\skel(G)$.

\emph{Case (ii) $N^-[y']\subset N^-[x']$:} By analogous arguments
as in Case $(i)$ one shows that the edge $x'y'$ 
connects $S^+(x)$ and $S^+(y)$ when $x'$ and $y'$ are chosen
such that $|N^-[y']-N^-[x']|$ becomes minimal. 

\emph{Case (iii):} 
Assume that neither Case $(i)$ nor $(ii)$ holds. Hence, if $x'y'$ with
$x'\in S^+(x)$ and $y'\in S^+(y)$ is dispensable by Condition $(D4)$,
then neither $(1^-)$ nor $(2^-)$ of the $N^-$-condition can be
fulfilled. Hence, $(3^-)$ with some vertex $z$ must hold, 
that is, $N^-[x']\cap N^-[y']
\subset N^-[z]\cap N^-[y']$ and $N^-[x']\cap N^-[y'] \subset
N^-[z]\cap N^-[x']$. 
W.l.o.g. assume that $x'$ and $y'$ are chosen, such
that $|N^-[x']\cap N^-[y']|$ becomes maximal among all such pairs $x'
\in S^+(x)$ and $y' \in S^+(y)$. By Condition $(D4)$ 
it holds that $N^+[x']=N^+[z]$ or $N^+[y]=N^+[z]$ and thus, $z \in S^+(x)$
or $z \in S^+(y)$. However, Condition $(3^-)$ is fulfilled, and thus
$|N^-[x']\cap N^-[z]|$ and $|N^-[z]\cap N^-[y']|$ are greater than
$|N^-[x']\cap N^-[y']|$, a contradiction to the choice of $x'$ and
$y'$. Hence, $x'y'$ is \emph{not} dispensable and thus, the vertices 
contained $S^+(x)$ and $S^+(y)$ cannot lie in different 
connected components of $\skel(G)$. 

By analogous arguments one shows, that  $x,y\in V$ are in
the same connected component of $\skel(G)$ if
 $N^-[x]\subset N^-[y]$. 
\end{proof}

\begin{prop}
If $G=(V,E)$ is thin and connected, then $\skel(G)$ is connected.
\label{lem:skelConnected}
\end{prop}

\begin{proof}
For each edge $xy\in E(G)$ define an integer
$$k_{xy}=|N^+[x]\cap N^+[y]|+|N^-[x]\cap N^-[y]|.$$

Assume for contradiction, that $x$ and $y$ are in different connected
components of $\skel(G)$. Hence, $xy$ must be dispensable. Take among all 
dispensable edges $xy\in E$, where $x$ and $y$ are in different components 
of $\skel(G)$ the ones that have largest value $k_{xy}$.

By the same arguments as in the proof of Lemma \ref{lem:subConnected}
the edge $xy$ cannot be dispensable by Condition $(D5)$, since then
there is a vertex $z_1\in S^+(x)$ and $z_1\in S^-(y)$ and 
by Lemma \ref{lem:cC}, there is 
a walk connecting $x$ and $y$ consisting of non-dispensable edges only
and thus $x$ and $y$ are in the same connected component of $\skel(G)$.

Moreover, if for $x$ and $y$ one of the Conditions $(1^+)$, $(2^+)$, 
$(1^-)$ or $(2^-)$ holds, then Lemma \ref{lem:subConnected} implies that
$x$ and $y$ are in the same connected component of $\skel(G)$. 

If $(D1)$ with  $(3^+)$ and $(3^-)$ is satisfied, then 
 $N^+[x]\cap N^+[y]\subset N^+[x]\cap N^+[z]$, 
 $N^-[x]\cap N^-[y]\subset N^-[x]\cap N^-[z]$, 
 $N^+[x]\cap N^+[y]\subset N^+[y]\cap N^+[z]$ and 
 $N^-[x]\cap N^-[y]\subset N^-[y]\cap N^-[z]$.
Note, by Remark \ref{rem:edge} there is an edge
$xz \in E$ or $zx \in E$, as well as, an edge 
$yz \in E$ or $zy \in E$.
 But then, 
 $k_{xz}>k_{xy}$ and $k_{yz}>k_{xy}$. Since $xy$ is chosen
among all dispensable edges where $x$ and $y$ are in different 
components that have maximal value $k_{xy}$ we can
conclude that  $x$ and $z$, resp., $y$ and $z$ 
are in the same connected component of $\skel(G)$ 
or that $xz$, resp., $yz$ are non-dispensable. 
Both cases lead to a contradiction, since then 
$x$ and $y$ would be connected by a walk in $\skel(G)$. 

If $(D2)$ holds, then in particular Condition $(3^+)$ 
and the weak $N^-$-condition holds with $z_1$.
Therefore,  
$N^+[x]\cap N^+[y] \subset N^+[x]\cap N^+[z_1]$, 
$N^+[x]\cap N^+[y] \subset N^+[y]\cap N^+[z_1]$,  
$N^-[x]\cap N^-[y] \subseteq N^-[x]\cap N^-[z_1]$ and 
$N^-[x]\cap N^-[y] \subseteq N^-[y]\cap N^-[z_1]$. 
By Remark \ref{rem:edge} there is an edge 
$xz_1 \in E$ or $z_1x \in E$, as well as, an edge 
$yz_1 \in E$ or $z_1y \in E$.
Again, $k_{xz_1}>k_{xy}$ and $k_{z_1y}> k_{xy}$. By analogous 
arguments as in the latter case, we obtain a contradiction.  

If $(D3)$ holds, then there is a vertex  $z\in V$ with 
$N^+[x]\cap N^+[y]\subset N^+[x]\cap N^+[z]$, 
$N^+[x]\cap N^+[y]\subset N^+[y]\cap N^+[z]$ and 
$N^-[x]=N^-[z]$ or $N^-[y]=N^-[z]$. 
If $N^-[x]=N^-[z]$, then Lemma \ref{lem:cC} implies that 
$x$ and $z$ are connected by a walk. 
Moreover,  for $zy$ holds then 
$N^-[x]\cap N^-[y]= N^-[y]\cap N^-[z]$ and still 
$N^+[x]\cap N^+[y]\subset N^+[y]\cap N^+[z]$.
Note, by Remark \ref{rem:edge} there is an edge 
$yz \in E$ or $zy \in E$.
Again, $k_{yz}>k_{xy}$ and by analogous arguments as before, 
$yz$ is connected by a walk in $\skel(G)$.
Combining the $xz$-walk and the $yz$-walk yields
a $xy$-walk in $\skel(G)$, a contradiction. 
Similarly, one treats the case when $N^-[y]=N^-[z]$. 
Analogously, one shows that Condition $(D4)$ leads to 
a contradiction. 

To summarize, for each dispensable edge $xy$ there is a walk
connecting $x$ and $y$ that consists 
of non-dispensable edges only, and thus
$\skel(G)$ is connected. 
\end{proof}

Since $\skel(G)$ is uniquely defined and in particular entirely 
in terms of the adjacency structure of $G$, we have the
following immediate consequence of the definition.
\begin{prop}
Any isomorphism $\varphi: G \rightarrow H$, as a map $V(G) \rightarrow V (H)$, is also an
isomorphism  $\varphi: \skel(G) \rightarrow \skel(H)$
\end{prop}

\section{Algorithms}

By Theorem \ref{thm:uPFD-strong}, every finite simple connected
digraph has a unique representation as a strong product of prime
digraphs, up to isomorphism and the order of the factors. We shortly
summarize the top-level control structure of the algorithm for the
computation of the PFD. We first compute for a given digraph $G$ the
Relation $S$ and its quotient graph $G/S$. By Lemma
\ref{lem:quotientthin} the digraph $G/S$ is thin and thus, the
Cartesian skeleton $\skel(G/S)$ is uniquely determined. The key idea
is then to find the PFD of $G/S$ w.r.t.\   the strong product, which is
achieved by computing the PFD its Cartesian skeleton $\skel(G/S)$ w.r.t.\  
the Cartesian product and to construct the prime factors of
$G/S$ using the information of the PFD of $\skel(G/S)$. Finally, the
prime factors of $G/S$ need to be checked and in some cases be
combined and modified, in order to determine the prime factors of the
digraph $G$ w.r.t.\  strong product, see Figure \ref{fig:quotient} and
\ref{fig:stronggeneral} for examples.


\begin{table}[t]
\begin{algorithm}[H]
\caption{\texttt{Cartesian Skeleton}} 
\label{alg:CartSk}
\begin{algorithmic}[1]
    \STATE \textbf{INPUT:} A connected thin digraph $G=(V,E)$;
		\FOR{each edge $xy\in E$ } 
				\STATE Check the dispensability conditions $(D1)-(D5)$.
		  			 Compute the set $D$ of dispensable edges in $H$;
		\ENDFOR
    \STATE $\skel(G)\gets (V,E\setminus D)$  \label{alg:rem}
    \STATE \textbf{OUTPUT:} The Cartesian skeleton $\skel(G)$;
\end{algorithmic}
\end{algorithm} 

\begin{algorithm}[H]
\caption{PFD of \emph{thin} digraphs w.r.t.\  $\boxtimes$}
\label{alg:factor_thin}
\begin{algorithmic}[1]
    \STATE \textbf{INPUT:} a connected \emph{S-thin} digraph $G$
	\STATE Compute the Cartesian skeleton $\skel(G)$ with Algorithm \ref{alg:CartSk}
	\STATE Compute the Cartesian PFD of $\skel(G) = \Box_{i \in I} H_i$  with the algorithm of Feigenbaum \cite{Fei86}
	\STATE Find all minimal subsets $J$ of $I$ such that the $H_J$-layers of  \label{alg:minset1}
					$H_J\Box H_{I\setminus J}$ where $H_J=\Box_{j\in J}	H_j$			
					and $H_{I\setminus J}=\Box_{j\in I\in J}	H_j$ correspond to layers
					of a factor of $G$ w.r.t.\  the strong product 			
    \STATE \textbf{OUTPUT:} The prime factors of $G$;
\end{algorithmic}
\end{algorithm}

\begin{algorithm}[H]
\caption{PFD of digraphs w.r.t.\  $\boxtimes$}
\label{alg:nonthin_factor}
\begin{algorithmic}[1]
    \STATE \textbf{INPUT:} a connected digraph $G$
	\STATE Compute $G=G'\boxtimes K_l$, where $G'$ has no nontrivial factor \label{alg:thin1}
		      isomorphic to a complete graph $K_r$; 
	\STATE Determine the prime factorization of $K_l$, that is, of $l$; \label{alg:thin2}
	\STATE Compute $H=G'/S$;
	\STATE Compute PFD and prime factors $H_1,\dots, H_n$ of $H$ with Algorithm \ref{alg:factor_thin}
	\STATE  By repeated application of Lemma \ref{lem:gcd}, find all minimal subsets $J$ \label{alg:minset2}
		   of $I = \{1, 2, \dots, n\}$ such that there are graphs $A$ and $B$ 
		   with $G =A\boxtimes B$, $A/S  = \boxtimes_{i\in J} H_i$ and
		   $B  = \boxtimes_{j\in J\setminus I} H_j$. 
			Save $A$ as prime factor.
    \STATE \textbf{OUTPUT:} The prime factors of $G$;
\end{algorithmic} 
\end{algorithm}
\end{table}

We explain in the following the details of this approach more precise. We start
with the construction of the Cartesian skeleton (Algorithm \ref{alg:CartSk}) 
and the computation of the PFD
w.r.t.\  the Cartesian product of digraphs in
Section \ref{subsec:alg:Cart}. We continue to give algorithms for 
determining the prime factors of thin digraphs (Algorithm \ref{alg:factor_thin})
in Section \ref{subsec:alg:thin}
and non-thin digraphs (Algorithm \ref{alg:nonthin_factor}) in
Section \ref{subsec:alg:nonthin}.
Note, these algorithms are only simple generalizations of the Algorithms
24.3, 24.6 and 24.7 for undirected graphs given in \cite{Hammack:2011a}. The novel
improvement for directed graphs is the unique construction of the Cartesian skeleton.
Therefore, we will refer in most parts of the upcoming proofs to  results
established in \cite{Hammack:2011a} rather than to replicate the proofs.

\subsection{Algorithmic Construction of $\skel(G)$ and the PFD of Digraphs w.r.t.\  $\Box$}
\label{subsec:alg:Cart}

\begin{prop}
	For a thin connected digraph $G=(V,E)$ with maximum degree $\Delta$, 
	Algorithm \ref{alg:CartSk} computes
	the Cartesian skeleton $\skel(G)$ in $O(|E|\Delta^3)$ time.  
	\label{prop:alg:CartSk}
\end{prop}
\begin{proof}
	By the arguments given in Section \ref{sec:CartSk} the Algorithm is correct.
	
	To determine the time complexity, note that for any edge $xy\in E$
	one of the Conditions $(D1) - (D5)$ is satisfied for some $z_1, z_2\in
	V$ if $z_1,z_2\in N^-[x]\cup N^+[x]$ (and also $z_1,z_2\in
	N^-[y]\cup N^+[y]$), see Remark \ref{rem:edge}. This implies that there are at most
	$O(|E|\Delta)$ dispensability checks to do. Moreover, several
	intersection and subset relations of neighborhoods have to computed,
	which can all be done in $O(\Delta^2)$ since each neighborhood
	contains at most $\Delta$ elements. Hence, the overall time
	complexity of the for-loop is $O(|E|\Delta^3)$, while Line
	\ref{alg:rem} can be executed in constant time.
\end{proof}



The PFD of a connected digraph $G=(V,E)$ w.r.t.\  the Cartesian product is unique and 
can be computed in $O(|V|^2\log_2(|V|)^2)$ time, see the work of Feigenbaum \cite{Fei86}. The
 algorithm of Feigenbaum 
works as follows. First one computes the PFD w.r.t.\  
Cartesian product of the underlying undirected graph. 
This can be done with the Algorithm of Imrich and Peterin in 
$O(|E|)$ time, \cite{impe-2007}. It is then checked whether
there is a conflict in the directions of the edges between
adjacent copies of the factors, which also determines the
overall time complexity. If there is some conflict, then
different factors, need to combined. The latter step is repeated until
 no conflict exists. 

\begin{prop}[\cite{Fei86}]
	For a connected digraph $G=(V,E)$ the algorithm of Feigenbaum 
	computes the PFD of $G$ w.r.t.\  the Cartesian product in $O(|V|^2(\log_2|V|)^2)$ time.
	\label{prop:alg:Cart}
\end{prop}

\subsection{Factoring thin Digraphs w.r.t.\  $\boxtimes$}
\label{subsec:alg:thin}

We are now interested in an algorithmic approach for determining the PFD of connected
thin digraphs w.r.t.\  strong product, which works as follows.
For a given thin connected digraph $G$ one first computes the unique 
Cartesian skeleton $\skel(G)$.
This Cartesian skeleton is afterwards factorized with
the algorithm of Feigenbaum \cite{Fei86} and one obtains the
Cartesian prime factors of $\skel(G)$. Note, for an arbitrary factorization
$G=G_1\boxtimes G_2$ of a thin digraph $G$, Proposition \ref{prop:skelProd}
asserts that $\skel(G_1\boxtimes G_2) =\skel(G_1)\Box \skel(G_2)$. 
Since $\skel(G_i)$ is a spanning graph of
$G_i$, $i=1,2$, it follows that the $\skel(G_i)$-layers of $\skel(G_1)\Box
\skel(G_2)$ have the same vertex sets as the $G_i$-layers of $G_1\boxtimes
G_2$. Moreover, if $\boxtimes_{i\in I} G_i$ is the unique PFD of $G$ then
we have $\skel(G)=\Box_{i\in I} \skel(G_i)$. Since $\skel(G_i)$, $i\in I$
need not to be prime with respect to the Cartesian product, we can infer
that the number of Cartesian prime factors of $\skel(G)$, can be larger
than the number of the strong prime factors. Hence, given the PFD of
$\skel(G)$ it might be necessary to combine several Cartesian factors to
get the strong prime factors of $G$. These steps for computing the PFD 
w.r.t.\ the strong product of a thin digraph are
summarized in Algorithm \ref{alg:factor_thin}.

\label{subsec:alg:nonthin}
\begin{figure}[tbp]
  \centering
	\includegraphics[bb= 63 449 520 626, scale=0.7]{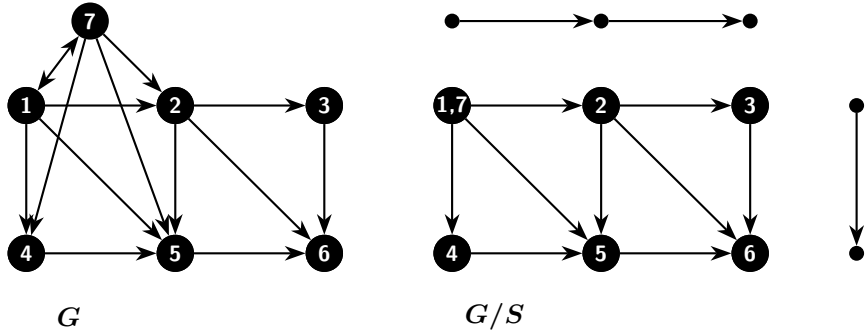}
	\caption{The digraph $G$ is prime. However, the quotient graph $G/S$ has a non-trivial
						product structure. Hence, the prime factors of $G/S$ must be combined, in order
            to find the prime factors of $G$.}
	\label{fig:quotient}
\end{figure}

\begin{prop}
	For a thin connected digraph $G=(V,E)$ with maximum degree $\Delta$, 
	Algorithm \ref{alg:factor_thin} computes the PFD of $G$
	 in $O(|V|^2(\log_2|V|)^2\Delta + |E|\Delta^3)$ time.  
	\label{prop:alg:factor_thin}
\end{prop}
\begin{proof}
	Note, Algorithm \ref{alg:factor_thin} is a one-to-one analog of 
  the algorithm for the PFD of undirected thin graphs, see
	\cite[Alg. 24.6]{Hammack:2011a}. 
	 The proof of correctness
	in \cite[Thm 24.9]{Hammack:2011a} for undirected graphs depends on the analogue 
	of Lemma \ref{lem:thinFactors}  
	and the unique construction of the Cartesian skeleton $\skel(G)$
	for the undirected case.
	Thus, using analogous arguments for directed graphs as in 
	\cite[Section 24.3]{Hammack:2011a}
  we can conclude  that Algorithm \ref{alg:factor_thin} is correct. 

	For the time complexity, observe that the Cartesian skeleton $\skel(G)$
	can be computed in $O(|E|\Delta^3)$ time and the PFD of $\skel(G)$
	in $O(|V|^2(\log_2|V|)^2)$ time. We are left with Line \ref{alg:minset1}
	and refer to \cite[Section 24.3]{Hammack:2011a}, where the time complexity
	of this step is determined with $O(|E||V|\log_2|V|)$.
	Since $|E|\leq  |V|\Delta$, we can conclude 
	that $|E||V|\log_2|V|\leq |V|^2\log_2|V|\Delta$. Thus, 
	we end in overall time complexity of $O(|V|^2(\log_2|V|)^2\Delta + |E|\Delta^3)$.
\end{proof}

\subsection{Factoring non-thin Digraphs w.r.t.\  $\boxtimes$}

\begin{figure}[tbp]
  \centering
	\includegraphics[bb=124 452 473 694, scale=1.1]{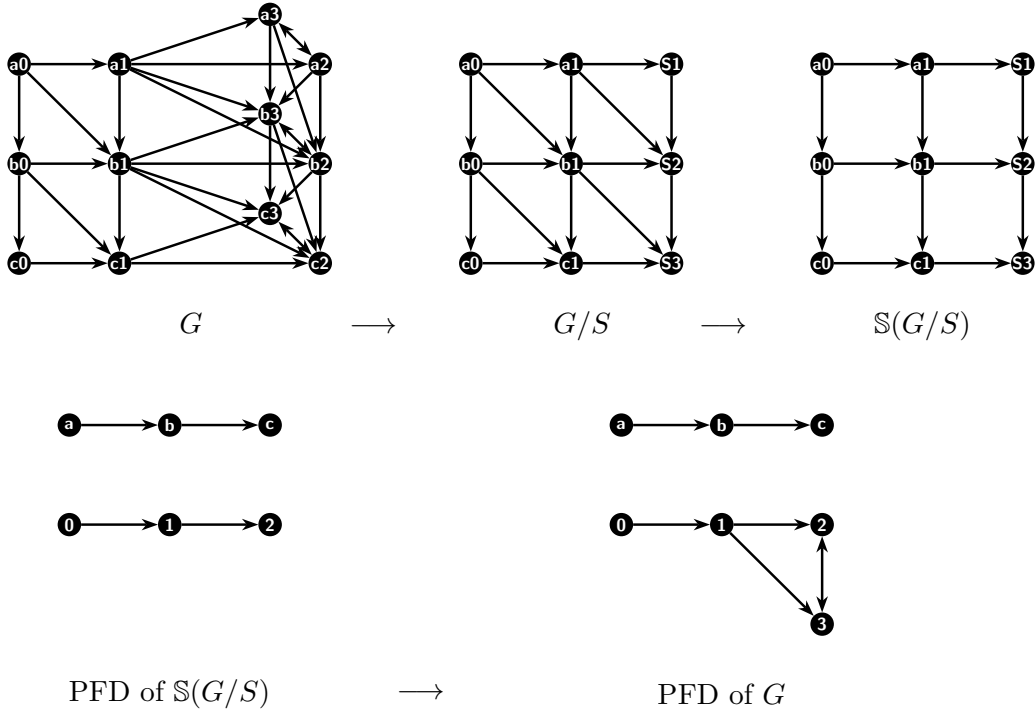}
	\caption{Illustrated are the basic steps of the PFD of strong product of digraphs.} 
	\label{fig:stronggeneral}
\end{figure}

We are now interested in an algorithmic approach for determining the PFD of connected
non-thin digraphs w.r.t.\  strong product, which works as follows.
Given an arbitrary digraph $G$, one first 
extracts a possible complete factor $K_l$ of maximal size, resulting in a graph $G'$, i.e.,
$G\simeq G'\boxtimes K_l$, and 
computes the quotient graph $H = G'/S$. This graph $H$ is thin and the PFD of $H$ w.r.t.\ 
the strong product can be computed with Algorithm \ref{alg:factor_thin}. 
Finally, given the prime factors of $H$ it might be the case that factors 
need to be combined to determine the prime factors of $G'$,
 see Figure \ref{fig:quotient}. This can be achieved by 
repeated application of Lemma \ref{lem:gcd}. 
Since $G\simeq G'\boxtimes K_l$, we can conclude that
the prime factors of $G$ are then the prime factors of $G'$ together
with the complete factors $K_{p_1},\dots,K_{p_j}$,
where $p_1\dots p_j$ are the prime factors of the
integer $l$. 
 This approach is summarized in Algorithm 
\ref{alg:nonthin_factor}. For an illustrative example see 
 Figure \ref{fig:stronggeneral}.

\begin{lem} 
Suppose that it is known that a given digraph $G$ that does not
admit any complete graphs as a factor is a strong product 
graph $G_1\boxtimes G_2$, and suppose that
the decomposition $G/S = G_1/S \boxtimes G_2/S$ is known. 
Then $G_1$ and $G_2$ can be determined from $G$, $G_1/S$ and $G_2/S$.

In fact, if $D(x_1, x_2)$ denotes the size of the S-equivalence class of $G$ that
is mapped into $(x_1, x_2) \in G_1/S \boxtimes G_2/S$, then the size $D(x_1)$ 
of the equivalence class of $G_1$ mapped into $x_1 \in G_1/S$ is 
$gcd\{D(x_1, y) \mid y \in  V (G_2)\}$.
Analogously for $D(x_2)$.
	\label{lem:gcd}
\end{lem}
\begin{proof}
Invoking Lemma \ref{lem:thinFactors}, \ref{lem:quotientthin}, \ref{lem:eqcl-complete},
and \ref{lem:prodQuotients}, the assertion can be implied by the same arguments 
as in the proof for undirected graphs \cite[Lemma 5.40]{IMKL-00}. 
\end{proof}



\begin{prop}
	For a connected digraph $G=(V,E)$ with maximum degree $\Delta$, 
	Algorithm \ref{alg:nonthin_factor} computes the PFD of $G$
	 in  $O(|V|^2(\log_2|V|)^2\Delta + |E|\Delta^3)$ time.  
	\label{prop:alg:nonthin-factor}
\end{prop}
\begin{proof}
	Again note,  Algorithm \ref{alg:nonthin_factor} is a one-to-one analog of 
  the algorithm for the PFD of undirected thin graphs, see \cite[Alg. 24.7]{Hammack:2011a}.
  The proof of correctness
	in \cite[Thm 24.12]{Hammack:2011a} for undirected graphs depends on the analogue 
	of Lemma \ref{lem:thinFactors}, \ref{lem:quotientthin}, \ref{lem:eqcl-complete},
	\ref{lem:prodQuotients} and the correctness of 
	Algorithm \ref{alg:factor_thin} for the undirected case. 
	Thus, using analogous arguments for directed graphs as in 
	\cite[Section 24.3 and Thm. 24.12 ]{Hammack:2011a}, 
	we can conclude the correctness of Algorithm \ref{alg:nonthin_factor}.

	For the time complexity, to extract complete factors $K_l$, its PFD and the computation of
	the quotient graph $G/S$ we refer to \cite[Lemma 24.10]{Hammack:2011a} and conclude that
	Line \ref{alg:thin1}-\ref{alg:thin2} run in  $O(|E|)$ time. 
	The PFD of $G/S$ w.r.t.\  the strong product can be computed in 
	$O(|V|^2(\log_2|V|)^2\Delta + |E|\Delta^3)$ time. 
	We are left with Line \ref{alg:minset2}
	and again refer to \cite[Section 24.3]{Hammack:2011a}, where the time complexity
	of this step is determined with $O(|E||V|\log_2|V|)$.
\end{proof}

\section{Summary and Outlook}

We presented in this paper the first polynomial-time algorithm that computes the prime factors 
of digraphs. The key idea for this algorithm was the construction of a unique Cartesian skeleton
for digraphs. The PFD of the Cartesian skeleton w.r.t.\  the Cartesian product
was utilized to find the PFD w.r.t.\  the strong product of the digraph under investigation. 

Since the strong product of digraphs is a special case of the so-called direct product
of digraphs, we assume that this approach can also be used to extend the known algorithms
for the PFD w.r.t.\  the direct product \cite{Kloeckl:07, Kloeckl:10}. The main challenge
in this context is a feasible construction of a so-called Boolean square, in which
the Cartesian skeleton is finally computed \cite{HAIM-09}. 

Moreover, we strongly assume that the definition of the Cartesian skeleton can be
generalized in a natural way for the computation of the strong product of di-hypergraphs
in a similar way as for undirected hypergraphs in \cite{HON-13, HOS11}.

Finally, since many graphs are prime although they can have 
a product-like structure, also known as \emph{approximate} graph products, 
the aim is to design algorithms that can handle such ``noisy'' graphs.  
Most of the  practically viable approaches are based 
 on \emph{local} factorization algorithms, that cover a graph by factorizable small patches 
and attempt to stepwisely extend regions with product structures  \cite{HIKS-08, HIKS-09, Hel-10, HIK-13, HIK-13b}. 
Since the construction of the Cartesian 
skeleton works on a rather local level, i.e, the usage of neighborhoods, we suppose that
our approach can in addition be used to establish local methods for finding 
approximate strong products of digraphs.


\bibliographystyle{plain}
\bibliography{biblio}

\end{document}